\documentclass[12pt]{article}
\pdfoutput=1
\usepackage[a4paper,margin=2cm,includefoot]{geometry}

\usepackage{amsmath,amsthm,amssymb,amsfonts}
\usepackage{thmtools,mathtools}

\usepackage{url}
\usepackage{hyperref}
\usepackage{bookmark}

\usepackage[shortlabels]{enumitem}
\setlist[enumerate,1]{label={(\arabic*)}}

\usepackage[nameinlink,capitalize]{cleveref}
\declaretheorem[name=Theorem]{thm}
\declaretheorem[name=Lemma,sibling=thm]{lma}
\declaretheorem[name=Corollary,sibling=thm]{col}
\declaretheorem[name=Definition,sibling=thm]{dfn}

\declaretheorem[name=Proposition,sibling=thm]{pro}

\usepackage{physics}

\usepackage[style=alphabetic,backref=true]{biblatex}
\DefineBibliographyStrings{english}{%
  backrefpage = {page},%
  backrefpages = {pages},%
}

\numberwithin{equation}{section}
\let\[\relax \let\]\relax %
\DeclareRobustCommand{\[}{\begin{equation}}
\DeclareRobustCommand{\]}{\end{equation}}

\usepackage[font=small]{caption}
\usepackage{subcaption}
\usepackage{tikz,float}
\usetikzlibrary{decorations.pathmorphing}
\usetikzlibrary{decorations.pathreplacing}
\usetikzlibrary{calc}
\usetikzlibrary{matrix}

\usepackage[affil-it]{authblk}

\hypersetup{
  pdflang=en,
  pdfpagelayout=OneColumn,
  pdftitle={A sublinear time quantum algorithm for longest common substring problem between run-length encoded strings},
  pdfauthor={Tzu-Ching Lee},
  hidelinks,
}

\title{A sublinear time quantum algorithm for longest common substring problem between run-length encoded strings}
\author{Tzu-Ching Lee}
\author{Han-Hsuan Lin}
\affil{Department of Computer Science, National Tsing Hua University, Hsinchu 30013, Taiwan}
\date{}

\usepackage{xcolor}
\newcommand{\wtilde}[1]{\widetilde{#1}}
\newcommand{\wt}[1]{\widetilde{#1}}

\newcommand{\afa}{\alpha}
\newcommand{\bta}{\beta}
\newcommand{\gma}{\gamma}

\newcommand{\Omg}{\Omega}
\newcommand{\dta}{\delta}

\newcommand{\eps}{\epsilon}

\newcommand{\cA}{\mathcal{A}}
\newcommand{\cB}{\mathcal{B}}

\newcommand{\cO}{\mathcal{O}}

\newcommand{\Od}{\cO}
\newcommand{\tO}{\tilde{\Od{}}}
\newcommand{\tOmg}{\tilde{\Omg{}}}

\newcommand{\ts}{\tilde{s}}
\newcommand{\tS}{\tilde{S}}
\newcommand{\tA}{\tilde{A}}
\newcommand{\tB}{\tilde{B}}
\newcommand{\TS}{\tilde{S}}%
\newcommand{\tn}{\tilde{n}}
\newcommand{\ti}{\tilde{i}}
\newcommand{\tj}{\tilde{j}}
\newcommand{\td}{\tilde{d}}
\newcommand{\Tt}{\tilde{t}}%

\newcommand{\PARITY}{\mathsf{PARITY}}
\newcommand{\LCS}{\mathsf{LCS}}
\newcommand{\LCSRLE}{\mathsf{LCS\text{-}RLE}}
\newcommand{\LCSRLEP}{\LCSRLE^\mathsf{p}}
\newcommand{\ELLCSRLE}{\mathsf{EL\text{-}LCS\text{-}RLE}}
\newcommand{\DLLCSRLE}{\mathsf{DL\text{-}LCS\text{-}RLE}}

\newcommand{\calA}{\mathcal{A}}

\newcommand{\Piv}{P^{\text{ -1}}}

\newcommand{\polylog}{\mathrm{polylog}}
\newcommand{\poly}{\mathrm{poly}}

\DeclarePairedDelimiter\set{\{}{\}}
\DeclarePairedDelimiter\PP{(}{)}

\DeclarePairedDelimiter\floor{\lfloor}{\rfloor}
\DeclarePairedDelimiter\ceil{\lceil}{\rceil}

\newcommand{\PH}{{}\cdot{}}

\newcommand{\Ie}{I.e.\ }
\newcommand{\ie}{i.e.\ }
\newcommand{\eg}{e.g.,\ }

\newcommand{\wrt}{w.r.t.\ }

\usepackage{expl3}

\begin{document}
\def\included{yep}
\tikzset {
	wavy/.style={ decoration={coil,aspect=0,amplitude=1pt} },
	desc/.style={ node font=\footnotesize },
	brkt/.style={ pos=0.25 },
}

\ExplSyntaxOn
\NewDocumentCommand{\rle}{ m } {
	\group_begin:
	\bool_set_false:N \l_tmpa_bool
	\clist_map_inline:nn {#1} {
		\bool_set_inverse:N \l_tmpa_bool
		\bool_if:NTF \l_tmpa_bool	{
			\tl_set:Nn \l_tmpa_tl {##1}
		} {
			\texttt{\l_tmpa_tl}^{##1}
		}
	}
	\group_end:
}

\NewDocumentCommand{\torle}{ m } {
	\group_begin:
	\str_clear:N \l_tmpa_str
		\str_map_inline:nn {#1} {
			\str_if_in:NnTF \l_tmpa_str {##1} {
				\str_put_right:Nn \l_tmpa_str {##1}
			} {
				\__rle_str_run:N \l_tmpa_str
				\str_set:Nn \l_tmpa_str {##1}
			}
		}
	\__rle_str_run:N \l_tmpa_str
	\group_end:
}

\cs_new:Nn \__rle_str_run:N {
	\str_if_empty:NF #1 {
		\texttt{\str_head:N #1}^{\str_count:N #1}
	}
}

\cs_new:Nn \__rle_parse:nNN {
	\str_map_inline:nn {#1} {
		\str_if_in:nnTF {1234567890} {##1} {
			\seq_put_right:Nn #3 {##1}
		} {
			\seq_put_right:Nn #2 {##1}
		}
	}
}

\cs_new:Nn \__rle_plus_fold:NNnn {
	\seq_map_indexed_inline:Nn #2 {
		\int_compare:nNnT {##1} > {#4} {
			\seq_map_break:
		}
		\int_compare:nNnF {##1} < {#3} {
			\int_add:Nn #1 {##2}
		}
	}
}

\cs_new:Nn \__rle_path:nn {
	\str_if_eq:nnTF {#1} {notail} {
		\path[save~path=\rlepath] (0, 0) rectangle (#2, \h);
	} {
		\path[save~path=\rlepath]
			(#2-\h/8+2, 0) decorate[wavy] {-- ++(0, \h) }
			-- (0, \h) |- cycle
			(#2+\h/8+2, 0) decorate[wavy] {-- ++(0, \h) }
			-- ++(2-\h/8, 0) |- cycle;
	}
}

\cs_new:Nn \__rle_run_text:nn {
	\node at (\l_tmpa_int+#2*0.5,\h/2) {\texttt{#1}$^{#2}$};
	\int_add:Nn \l_tmpa_int {#2}
	\draw (\l_tmpa_int,0) -- ++(0, \h);
}

\int_new:N \l_a_len_int
\int_new:N \l_h_len_int
\int_new:N \l_t_len_int
\int_new:N \l_r_len_int
\int_new:N \l_b_len_int
\int_new:N \l_r_cnt_int

\NewDocumentCommand{\drawrle}{ O{tail} m O{1} O{10000} o } {
	\group_begin:
	\__rle_parse:nNN {#2}
		\l_tmpa_seq \l_tmpb_seq

	\int_set:Nn \l_r_cnt_int {\seq_count:N \l_tmpb_seq}

	\__rle_plus_fold:NNnn
		\l_a_len_int \l_tmpb_seq {1} {\l_r_cnt_int}

	\__rle_plus_fold:NNnn
		\l_h_len_int \l_tmpb_seq {1} {#3}

	\__rle_plus_fold:NNnn
		\l_t_len_int \l_tmpb_seq {#3+1} {\l_r_cnt_int}

	\__rle_plus_fold:NNnn
		\l_r_len_int \l_tmpb_seq {#3} {#3}
	\int_decr:N \l_r_len_int

  \str_if_eq:nnTF {#5} {-NoValue-} {
		\__rle_plus_fold:NNnn
			\l_b_len_int \l_tmpb_seq {#3+1} {\int_min:nn {#4} {\l_r_cnt_int}}
		\int_incr:N \l_b_len_int
		\int_compare:nNnT {#4} > {\l_r_cnt_int} {
			\str_if_eq:nnF {#1} {notail} {
				\int_add:Nn \l_b_len_int {4}
			}
		}
	} {
		\int_set:Nn \l_b_len_int {#5-\l_r_len_int-1}
		\int_compare:nNnTF {\l_b_len_int+\l_h_len_int} > {\l_a_len_int} {
				\int_add:Nn \l_b_len_int {4}
		} {
				\int_add:Nn \l_b_len_int {1}
		}
	}

	\begin{scope}[xscale=\xs, shift={(-\l_h_len_int, 0)}]

		\__rle_path:nn {#1} {\l_a_len_int}

		\coordinate (HEAD) at (0, \h/2);
		\coordinate (R) at (\l_h_len_int-1-\l_r_len_int, 0);
		\coordinate (B) at (\l_h_len_int-1+\l_b_len_int, 0);

		\begin{scope}
			\clip[use~path=\rlepath];
			\fill[fill=red!30]  (\l_h_len_int-1, 0) rectangle ++(-\l_r_len_int, \h);
			\fill[fill=cyan!30] (\l_h_len_int-1, 0) rectangle ++( \l_b_len_int, \h);
		\end{scope}

		\draw[dashed] (\l_h_len_int-1, 0) -- ++(0, \h);
		\draw[thick, use~path=\rlepath];

		\int_set:Nn \l_tmpa_int {0}
		\int_step_inline:nnn {1} {\l_r_cnt_int} {
			\__rle_run_text:nn
				{\seq_item:Nn \l_tmpa_seq {##1}}
				{\seq_item:Nn \l_tmpb_seq {##1}}
		}
		\str_if_eq:nnF {#1} {notail} {
			\node at (\l_tmpa_int+2/2, \h/2) {$\cdots$};
			\node at (\l_tmpa_int+2+2/2, \h/2) {$\cdots$};
		}
	\end{scope}
	\group_end:
}
\ExplSyntaxOff

\maketitle

\begin{abstract}
	We give a sublinear quantum algorithm for the longest common substring (LCS) problem on the run-length encoded (RLE) inputs, under the assumption that the prefix-sums of the runs are given.
	Our algorithm costs $\tO(n^{5/6})\cdot\Od(\polylog(\tn))$ time, where $n$ and $\tn$ are the encoded and decoded length of the inputs, respectively.
	We justify the use of the prefix-sum oracles by showing that, without the oracles, there is a $\Omg(n/\log^2n)$ lower-bound on the quantum query complexity of finding LCS given two RLE strings due to a reduction of $\PARITY$ to the problem.
\end{abstract}

\section{Introduction}

String processing is an important field of research in theoretical computer science.
There are lots of results for various classic string processing problems such as string matching \autocite{KnuthMP77, BoyerM77, KarpR87}, longest common substring, edit distance.
The development of string processing algorithms has led to the discovery of many impactful computer science concepts and useful tools, including dynamic programming, suffix tree \autocite{Weiner73, Farach97} and trie \autocite{Fredkin60}.
String processing also has applications in various fields such as bioinformatics \autocite{NeedlemanW70}, image analysis \autocite{HindsFD90}, and compression \autocite{ZivL77}.

A natural extension of string processing is to do it between \emph{compressed strings}. Ideally, the time cost of the string processing between compressed strings would be independent of the decoded lengths of the strings. Since the compressed string can be much shorter than the original string, this would significantly save computation time. Of course, whether such fast string processing is possible depends on what kind of compression scheme we are using.

Run-Length Encoding (RLE) is a simple way to compress strings. In RLE, the consecutive repetition of a character (run) is replaced by a character-length pair -- the character itself and the length of run.
For example, the RLE of the string \texttt{aaabcccdd} is $\rle{a,3,b,1,c,3,d,2}$.
RLE is a common method to compress fax data \autocite{ITU-T.4:2004}, and is a part of the JPEG and TIFF image standard \autocite{ISO10918:1994, ISO12639:1998} as well.
String processing on RLE strings has been studied a lot. Apostolico, Landau, and Skiena gave an $\Od(n^2\log n)$-time algorithm on finding Longest Common Subsequence between two RLE compressed strings \autocite{ApostolicoLS99}, where $n$ is the length of the compressed strings.
Hooshmand, Tavakoli, Abedin, and Thankachan obtained an $\Od(n\log n)$-time algorithm on computing the Average Common Substring with RLE inputs \autocite{HooshmandTAT18}.
Chen and Chao proposed an algorithm to compute the edit distance between two RLE strings \autocite{ChenC13}, and the result was further improved to near-optimal by Clifford, Gawrychowski, Kociumaka, Martin and Uznanski in \autocite{CliffordGKMU19}, which runs in $\Od(n^2\log n)$ time.

Alternatively, another way to speedup string processing is to use \emph{quantum algorithms}. If we use {quantum algorithms} for string processing, it is possible to get time cost \emph{sublinear} in the input length because the quantum computer can read the strings in superposition. One of the earliest such results was by Hariharan and Vinay \autocite{HariharanV03}, who constructed an $\tO(\sqrt{n})$-time string matching quantum algorithm, in which Grover's search \autocite{Grover96} and Vishkin's deterministic sampling technique \autocite{Vishkin90} were used to reach this near-optimal time complexity.
Recently, Le Gall and Seddighin \autocite{LeGallS22} obtained several sublinear-time quantum algorithms for various string problems, including an $\tO(n^{5/6})$-time algorithm for longest common substring (LCS), an $\tO(\sqrt{n})$-time algorithm for longest palindrome substring (LPS), and an $\tO(\sqrt{n})$-time algorithm for approximating ulam distance.
Another recent work is done by Akmal and Jin \autocite{AkmalJ22}, using synchronizing sets \autocite{KempaK19} with quantum walk \autocite{MagniezNRS11}, showing that LCS can be solved in $\tO(n^{2/3})$ quantum time. They also introduced a $n^{1/2+o(1)}$ algorithm for lexicographically minimal string rotation problem, and a $\tO(\sqrt{n})$-time algorithm for longest square substring problem in the same paper.

In this work, we combine the two above ideas and investigate the possibility of using \emph{quantum algorithm} to do string processing on \emph{compressed strings}, while keeping the advantages of both methods. Thus, we ask the following question:
\begin{center}
Is it possible to have a quantum string processing algorithm whose time cost is sublinear in the encoded lengths of the strings and independent of the decoded lengths?\footnote{With a non-trivial string problem and a non-trivial compression scheme.}
\end{center}

  The main result of this paper is the first almost\footnote{We have $\polylog$ dependence on the decoded length $\tn$.} affirmative answer to the above question:
 
\begin{thm}[Informal]
There is a quantum algorithm that finds the RLE of a LCS given two RLE strings in $\Od(n^{5/6}\cdot\polylog(n)\cdot\polylog(\tn))$ time, with oracle access to the RLE strings and the prefix-sum of their runs\footnote{prefix-sum is defined in \autoref{dfn:prefix-sum}. We justify the use of prefix-sum oracle with the following two facts: first, constructing it only adds a constant factor in preprocessing time. Second, without the prefix-sum oracles, finding an LCS from RLE inputs needs at least $\tOmg(n)$ queries due to a reduction of the $\PARITY$ problem.}, where $n$ and $\tn$ are the encoded length and the decoded length of the inputs, respectively.
\end{thm}

Our algorithm mainly follows the construction of the LCS algorithm of \autocite{LeGallS22}, but we need to overcome a big difficulty -- the longest common substring in terms of encoded length may differ from the longest common substring in terms of decoded length. For example, between $\torle{abcdbbbbccccc}$ and $\torle{abcd@bbbbcc}$, the RLE string $\torle{abcd}$ is the longest one in terms of encoded length, and $\torle{bbbbcc}$ is the longest one in terms of decoded length, which is what we want to find.
So applying existing LCS algorithms for strings directly on RLE inputs is not an option. We deal with this problem with various techniques and observations in the algorithm.

\subsection{Related work}

Recently, Gibney and Thankachan \autocite{GibneyT23} developed a $\tO(\sqrt{zn})$-time quantum algorithm for the Lempel-Ziv77 algorithm (LZ77) \autocite{ZivL77} and calculating the Run-length-encoded Burrows-Wheeler Transform (RL-BWT), where $n$ is the length of the input string and $z$ is the number of factor in the LZ77 factorization of the input, which roughly corresponds to the encoded length of that string. Given two strings $A$ and $B$, they showed how to calculate the LZ77 compression and a supporting data structure of $A\$B$. With this compressed data, the LCS between $A$ and $B$ can be found efficiently. Note that this is different from our input model because they can do preprocessing on the \emph{concatenated string} $A\$B$, while in our work, $A$ and $B$ are preprocessed independently.

\subsection{Overview of the algorithm}

The LCS algorithm in \autocite{LeGallS22} starts with a binary search on the solution size $d$, in which, different algorithms are used to find an LCS of length at least $d$ for different sizes. For $d\leq n^{1/3}$, an element distinctness algorithm is used; for $d\geq n^{1/3}$, an amplitude-amplified sampling algorithm is used.

Here, to find the longest common substring from two RLE strings $A$ and $B$ of decoded length $\tn$, we do a binary search on the decoded length of the answer $\td\in[\tn]$. In each iteration of the binary search, we check whether a common substring $s$ of decode length $|\ts|=\td$ exists.
The checking algorithm consists of two sub-algorithms, one finds the ``short" answers, and the other finds the ``long" ones. We say a common substring $s$ is ``short" if its \emph{encoded length} is not greater than $n^{1/3}$, otherwise we say it is ``long". Different from \autocite{LeGallS22}, in our work, both sub-algorithms are executed in every iteration of the binary search.

\paragraph{finding short answers}
We use Ambainis' element distinctness between two sets of size $n$ to find an LCS from two RLE strings, in which we consider $i\in[n]$ and $j\in[n]$ being ``equal" if they correspond to the first run of a common substring $s$ in $A$ and $B$, respectively\footnote{This is possible because Ambainis' algorithm can solve problems more general than element distinctness. See e.g. \autocite{ChildsE05}}. More precisely, $i\in[n]$ and $j\in[n]$ are ``equal" if and only if $f(i,A,j,B)\vee f(j,B,i,A)=1$, where $f(i,A,j,B)$ is a predicate that is $1$ iff a common substring of decoded length $\td$ and encoded length at most $n^{1/3}$ starting from the beginning of $A[i]$ and within $B[j]$ exists.
Here we utilize the observation that, for the RLE of an LCS $s$, the first character of $s$ must correspond to either the beginning of a run in $A$, the beginning of a run in $B$, or both of them. $f(i,A,j,B)$ can be seen as finding a substring with the assumption that it starts from the beginning of $A[i]$. The separation of $f(i,A,j,B)$ and $f(j,B,i,A)$ enables us to use skip list and range minimum query to construct a data structure on $B$ that, for any given $i$, efficiently checks the existence of $j$, \ie calculates $f(i,A,j,B)$. We then calculate $f(j,B,i,A)$ with a data structure on $A$.
Element distinctness costs $\tO(n^{2/3}\cdot T)$ time where $T$ is the time to compare two elements.
Here two elements are compared by using Grover's search to find a mismatched run, with some extra $\tO(1)$-time checks. Hence the time complexity is $\tO(n^{2/3}\cdot \sqrt{n^{1/3}})=\tO(n^{5/6})$.

\paragraph{finding long answers}
Consider the following procedure: choose an integer $n^{1/3} <\ell\leq n$. Sample two RLE strings $P=A[i:i+2\floor{\ell/3}-1]$ and $S=B[j:j+\ell-1]$ with uniformly random $i,j$.
Suppose a common substring $s$ with encoded length $|s|=\ell$ exists, then the probability that $P$ being a substring of $s$ and $S$ covers the corresponding $P$ in $B$ at the same time, is $\Omg(\ell^2/n^2)$.
By first finding the corresponding $P$ in $S$ with pattern matching with time cost $\tO(\sqrt{\ell})$, we can grow the result into a common substring of encoded length $2\ell$ using minimal finding in $\tO(\sqrt{\ell})$ time.
This procedure succeeds with probability $\Omg(\ell^2/n^2)$.
With amplitude amplification, we can boost the success probability of the whole above procedure to $2/3$ and get a procedure that searches for an common substring of encoded length between $\ell$ and $2\ell$ with time complexity $\tO(\sqrt{n^2/\ell^2}\cdot \sqrt{\ell})=\tO(n/\sqrt{\ell})$.
We run the amplified procedure repeatedly with $\ell$ being set to $n^{1/3},2n^{1/3},4n^{1/3},\ldots$, until an answer is found or $\ell>n$. The search stops in $\Od(\log n)$ iterations, results in time complexity of $\Od(\log n)\cdot\tO(n/\sqrt{n^{1/3}})=\tO(n^{5/6})$.

\subsection{Paper organization}
In \autoref{sec:prelim}, we introduce the notations and the definitions.
The main algorithm is explained in \autoref{sec:meat} with the sub-algorithm for short answers in \autoref{sec:small-d} and the one for long answers in \autoref{sec:large-d}. In \autoref{sec:LB}, we investigate the query lower-bound of finding LCS from two RLE strings without prefix-sum oracles.

\section{Preliminaries}\label{sec:prelim}

\subsection{Conventions and Notations}
We abbreviate both ``run-length encoding" and ``run-length-encoded" to ``RLE".
We use tilde ( $\wt{\PH}$ ) to denote decoded strings and their properties, while notations without tilde refer to their RLE counterparts.
We use calligraphic letters (\eg $\cA$ and $\cB$) to denote algorithms, use teletype letters (\eg $\texttt{a}$) to denote strings or character literals, and use sans-serif letters (\eg $\LCS$) to denote problems.
We count indices from $1$.
By $[m]$, we mean the set $\set{1,2,\ldots,m}$.
The asymptotic notations $\tO(\PH)$ and $\tOmg(\PH)$ hide $\polylog(n)$ and $\polylog(\tn)$ factors, where $n$ is the encoded length of the input, and $\tn$ is the decoded length of the input. We say a quantum algorithm succeeds with high probability if its success probability is at least $\Omg(1-1/\poly(n))$.

\paragraph{Strings}
A \emph{string} $\ts\in\Sigma^\ast$ is a sequence of characters over a character set $\Sigma$.
The length of a string $\ts$ is denoted as $|\ts|$.
For a string $\ts$ of length $n$, a \emph{substring} of $\ts$ is defined as $\ts[i:j] := \ts[i':j'] = \ts[i']\ts[i'+1]\ldots\ts[j']$, where $i'=\max(1,i)$ and $j'=\min(n,j)$. \Ie it starts at the $i'$-th character and ends at the $j'$-th character.
If $i>j$, we define $s[i:j]$ as an empty string $\eps$.
We say a string $s$ is \emph{$\ell$-periodic} if $s[i]=s[i+\ell]$ for all $1\leq i\leq |s|-\ell$.

\paragraph{Run-length Encoding}
\emph{Run-length encoding} (RLE) of a string $\ts$, denoted as $s$, is a sequence of runs of identical characters $s[1]s[2]\cdots s[n]$, where $s[i]$ is a maximal run of identical characters, $n$ is the length of $s$, \ie the number of such runs.
For a run $s[i]$, $R(s[i])$ is its length and $C(s[i])$ denotes the unique character comprising the run.
When we write out $s$ explicitly, we write each $s[i]$ in the format of $C(s[i])^{R(s[i])}$, with $C(s[i])$ in a teletype font (\eg $\texttt{a}^3$).
Equivalently, each run $s[i]$ can be represented as a character-length pair $(C(s[i]), R(s[i]))$.
When there exists $i$ and $j\geq i$ such that $t = s[i:j]$, we call $t$ an \emph{substring} of $s$.
We say an RLE string $t$ is a \emph{generalized substring} of $s$ if the decoded string $\Tt$ is a substring of $\ts$.

\subsection{Computation Model}

\paragraph{Quantum Oracle}

Let $S$ be an RLE string. In a quantum algorithm, we access an RLE string $S$ via querying the oracle $O_S$. More precisely,
\[
	O_S:
	\ket{i}\ket{c}_{\text{char}}\ket{r}_{\text{run}}
	\mapsto
	\ket{i}\ket{c\oplus C(S[i])}_{\text{char}}\ket{r \oplus R(S[i])}_{\text{run}}
\]
is a unitary mapping for any $i\in[|S|]$, any $c\in\Sigma$, and any $r\in[\tn]$.
The corresponding prefix-sum $\Piv_S$ is accessed from the unitary mapping
\[
	O_P:
	\ket{i}\ket{x}
	\mapsto
	\ket{i}\ket{x\oplus P_S[i]},
\]
for any $i\in\set{0}\cup [|S|]$ and any $x\in[\tn]$.

\paragraph{Word RAM model}

We assume basic arithmetic and comparison operations between two bit strings of length $O(\log(\tn))$ and $O(\log n)$ all cost $O(1)$ quantum time.

\subsection{Primitives}\label{sec:primitive}

\paragraph{Grover's search (\autocite{Grover96}).}
Let $f:[n]\rightarrow\set{0,1}$ be a function. There is a quantum algorithm $\cA$ that finds an element $x\in[n]$ such that $f(x)=1$ or verifies the absence of such an element.
$\cA$ succeeds with probability at least $2/3$ and has time complexity $\tO(\sqrt{n}\cdot T(n))$, where $T(n)$ is the complexity of computing $f(i)$ for one $i\in[n]$.

\paragraph{Amplitude amplification (\autocite{BrassardH97}, \autocite{Grover98}).}
Let $\cA$ be a quantum algorithm that solves a decision problem with one-sided error and success with probability $p\in(0,1)$ in $T$ quantum time. There is another quantum algorithm $\cB$ that solves the same decision problem with one-sided error and success probability at least $2/3$ in $\tO(T/\sqrt{p})$ quantum time.

\paragraph{Minimum finding (\autocite{Durr96}).}
Let $f:[n]\rightarrow X$ be a function, where $X$ is a set with a total order. There is a quantum algorithm $\cA$ that finds an index $i\in[n]$ such that $f(i)\leq f(j)$ for all $j\in[n]$. $\cA$ succeeds with probability at least $2/3$ and costs $\tO(\sqrt{n}\cdot T)$ time, where $T$ is the time to compare $f(i)$ to $f(j)$ for any $i,j\in[n]$.

\paragraph{Element distinctness (\autocite{Ambainis07}, \autocite{LeGallS22})}\hspace{-1em}\footnotemark{}
Let $X$ and $Y$ be two lists of size $n$ and $f:(X\cup Y)\rightarrow\mathbb{N}$ be a function. There is a quantum algorithm $\cA$ that finds an $x\in X$ and a $y\in Y$ such that $f(x)=f(y)$. $\cA$ succeeds with probability at least $2/3$ and costs $\tO(n^{2/3}\cdot T(n))$ time; $T(n)$ is the time to do the three-way comparison between $f(a)$ and $f(b)$ for any $a,b\in X\cup Y$.

\footnotetext{The definition here is also known as claw finding. The time upper bound is obtained in \autocite[Section 2.1]{LeGallS22}. We also explain it in \autoref{sec:QW}}

\paragraph{Pattern matching (\autocite{HariharanV03}).}
Let $P$ and $S$ be two strings. There is a quantum algorithm $\cA$ that finds the leftmost (or the rightmost) instance of $P$ in $T$, or it verifies the absence of such an instance. $\cA$  succeeds with probability at least $2/3$ and costs $\tO(\sqrt{|P|}+\sqrt{|S|})$ quantum time.

\begin{lma}[Boost to high success probability]\label{thm:whp}
Let $\cA$ be a bounded-error quantum algorithm with time complexity $\Od(T)$.
By repeating $\cA$ for $\Od(\log n)$ times then output the majority of the outcomes, we can boost the success probability of $\cA$ to $\Omg(1-1/\poly(n))$ with overall time complexity $\Od(T\cdot\log n)$.
\end{lma}

\autoref{thm:whp} enables us to do Grover's search over the outcomes of applying $\cA$ on different inputs, because quantum computational errors accumulates linearly.\footnote{In fact, it is possible to apply Grover's search over bounded-error verifier \emph{without} the logarithmic overhead \autocite{HoyerMW03}.}

\subsection{Definitions}\label{sec:defs}

\begin{dfn}[Longest Common Substring (LCS)]
A string $\ts$ is a \emph{longest common substring} (LCS) of strings $\tA$ and $\tB$ if it is a substring of both, and $|\ts|\geq|\Tt|$ for every common substring.
\end{dfn}

\begin{dfn}[Generalized substring of an RLE String] \label{dfn:rle-substr}
For two RLE strings $s$ and $t$, we say $s$ is a \emph{generalized substring} of $t$ if $\ts$ is a substring of $\Tt$.
For example, for $t=\torle{aaabbbbccddddd}$, the RLE string $\torle{bbbbcc}$ is a substring as well as a generalized substring, but $\torle{abbbbccdd}$ is a generalized substring but not a substring.
\end{dfn}

\begin{dfn}[LCS Problem on RLE Strings]\label{dfn:LCSRLE}
 Given oracle access to two RLE strings $A$ and $B$, find the longest common generalized substring $s$ of $A$ and $B$, \ie the RLE of an LCS $\ts$ between $\tA$ and $\tB$, and locate an instance of $s$ in each input. More precisely, find a tuple $(i_A, i_B, \ell)$ such that $|s|=\ell$, with an instance of $s$ starts within the run $A[i_A]$, and another instance of $s$ starts within the run $B[i_B]$. We denote this problem as $\LCSRLE$.
\end{dfn}

\begin{dfn}[Decoded Length of LCS on RLE Strings Problem ($\DLLCSRLE$)]\label{dfn:DLLCSRLE}
Given oracle access to two RLE strings $A$ and $B$,
calculate $|\ts|$ such that $\ts$ is an LCS of $\tA$ and $\tB$.
We denote this problem as $\DLLCSRLE$.
\end{dfn}

\begin{dfn}[Encoded Length of LCS on RLE Strings Problem]\label{dfn:ELLCSRLE}
Given oracle access to two RLE strings $A$ and $B$,
calculate $|s|$ such that $\ts$ is an LCS of $\tA$ and $\tB$.
We denote this problem as $\ELLCSRLE$.
\end{dfn}

In \autoref{sec:LB} we show a near linear lower bound on query complexity for $\DLLCSRLE$,
and a near-linear, $\Omg(n/\log^2n)$, lower bound for $\ELLCSRLE$.
Both problems are bounded by reductions from the parity problem.

\begin{dfn}[Parity Problem]\label{dfn:PARITY}
	Given oracle access to a length-$n$ binary string $B\in\set{0,1}^n$, find $\bigoplus_{i=1}^n B_i$, the parity of $B$, where $\oplus$ is addition in $\mathbb{Z}_2$. We denote this problem as $\PARITY$.
\end{dfn}

With a short reduction, we show that $\ELLCSRLE$ and $\LCSRLE$ shares the same lower bound on query complexity (\autoref{thm:LB-LCSRLE}).
As a result, we loosen the requirement, and assume the oracle of prefix-sum of the inputs are also given.

\begin{dfn}[prefix-sum of the runs of an RLE string]\label{dfn:prefix-sum}
For an RLE string $s$, $P_s[i]$ is the $i$th \emph{prefix-sum} of the runs, \ie $P_s[i] := \sum_{j=1}^iR(s[j])$, with $P[0] := 0$. Intuitively, $P_s[i]$ is the index where $s[i]$, the $i$-th run in $s$, ends in the decoded string $\ts$. As a consequence, for $i\leq j$, the decoded length of $s[i:j]$ is $P_s[j] - P_s[i-1]$.
\end{dfn}

Note that the prefix-sum oracle can be constructed and stored in QRAM in linear time while doing the RLE compression, thus constructing it only adds a constant factor in preprocessing time. Our main algorithm solves the LCS problem with the prefix-sum oracle provided, formalized as below.

\begin{dfn}[LCS Problem on RLE Strings, with Prefix-sum Oracles]\label{dfn:LCSRLEP}
Given oracle access to two RLE strings $A$ and $B$ and prefix-sums of their runs, $P_A$ and $P_B$,
find an RLE string $s$, such that $\ts$ is an LCS of their decoded counterparts $\tA$ and $\tB$.
More precisely, the algorithm outputs the same triplet as the one for $\LCSRLE$.
We denote this problem as $\LCSRLEP$.
\end{dfn}

\section{LCS from two RLE strings with Prefix-sum Oracles}\label{sec:meat}

In this section, we describe how our algorithm for $\LCSRLEP$ works.
We first introduce a handy subroutine -- inverse prefix-sum -- that translates the indices for a decoded string $\ts$ into the ones for the RLE string $s$ in $\Od(\log n)$ queries to the prefix-sum oracle.
Then we move on to show the internals of the main algorithm.

\paragraph{Overview of the algorithm}
To find a longest common substring from two RLE strings $A$ and $B$, we do a binary search on the decoded length of the answer $\td\in[\tn]$. In each iteration of the binary search, we check whether a common substring $s$ such that its decode length $|\ts|=\td$ exists.
The checking algorithm consists of two sub-algorithms, one finds the short answers, and the other one finds the long ones. Both of them are executed in every iteration of the binary search. Here, we say a common substring $s$ is ``short" if its \emph{encoded length} is not greater than $n^{1/3}$, otherwise we say it is ``long". The algorithm for short answers is discussed in \autoref{sec:small-d}, and the one for long answers is discussed in \autoref{sec:large-d}.

\subsection{Inverse Prefix-Sum}\label{sec:invPS}

In \autoref{sec:defs}, we introduced the notion of prefix-sum,
which can be seen as a translation of indices $i$ of an RLE string to the corresponding indices $\ti$ in the decoded string. Here we show how to do an inverse-ish translation, that sends $\ti$ to $i$.
This translation is utilized heavily in our algorithm.

\begin{lma}[Inverse Prefix-sum, $\Piv_S$]\label{thm:inv-prefix}
Given a prefix-sum oracle of an RLE string $S$ of encoded length $\Od(n)$, one can calculate the function $\Piv_S:[\tn]\rightarrow[n]$ that maps indices of $\TS$, the decoded string, to the corresponding ones of $S$ in $\Od(\log n)$ time.
\end{lma}
\begin{proof}
	Let $\ti\in[\tn]$ be a decoded index. To find the corresponding index $i\in[n]$, we do a binary search over $[n]$ to find the $i\in[n]$ such that $P_S[i-1]<\ti\leq P_S[i]$. The process is correct since the prefix-sum is strictly-increasing.
\end{proof}

\paragraph{A na\"ive approach}
The prefix-sum oracle, $\Piv_S$, together with the oracle to the RLE string, $S$, are sufficient to build the oracle for the corresponding decoded string, since a query to $\tS[\ti]$ can be answered with $C(S[P^{-1}_S(\ti)])$.
As a consequence, $\LCSRLEP$ can be solved using algorithms for LCS problem on decoded strings, such as \autocite{LeGallS22}, \autocite{AkmalJ22}, and \autocite{GibneyT23}, with an extra logarithmic factor of the encoded length to the time complexity.
For example, applying this approach on \autocite{LeGallS22} gives a $\tO(\tn^{5/6})$-time algorithm.
Note that this method does not utilize the fact that RLE strings could be much shorter than the decoded ones.
And that is what we take advantage of in our algorithm, whose time cost is $\tO(n^{5/6})$.

\subsection{Algorithm for short answers}\label{sec:small-d}

The algorithm for short (shorter than $n^{1/3}$) answers finds a common substring $s$ from $A$ and $B$ that has a decoded length $|\ts|$ at least $\td$ and has an encoded length $|s|$ not greater than $n^{1/3}$.
$s$ can be identified by a 3-tuple $(i_A,i_B,\ell)$ such that $s$ starts within the runs $A[i_A]$ and $B[i_B]$ in $A$ and $B$, respectively; and $\ell$ is the encoded length of $s$.
As will be shown later, with inverse prefix-sum oracles, $\ell$ can be derived from $i_A$ and $i_B$, so the algorithm only has to find $(i_A,i_B)$, which will be called \emph{$\td$-witness pair} in the following text.

\begin{dfn}[$\td$-witness pair]\label{dfn:witness-pair}
	For two RLE strings $A$ and $B$ of length $n$, and a positive integer $\td\leq\tn$, a pair $(i_A, i_B) \in [n]\times[n]$ is a \emph{$\td$-witness pair} if and only if there exists a common generalized substring $s$ between $A$ and $B$ such that $|\ts|\geq\td$, $|s|<n^{1/3}$, and it starts within runs $A[i_A]$ and $B[i_B]$ in $A$ and $B$, respectively.
\end{dfn}

To cover all possible cases, in \autoref{dfn:witness-pair}, the associated common substring can start from any character within the runs $A[i_A]$ and $B[i_B]$.
But it's possible to reduce the number of cases to check with the observation that, for a pair $(i_A,i_B)\in[n]\times[n]$ it must be the case that $R(A[i_A]) \leq R(B[i_B])$ or $R(A[i_A]) \geq R(B[i_B])$.

\begin{pro}\label{thm:shift}
	For $(i_A,i_B)\in[n]\times[n]$, $(i_A,i_B)$ is a $\td$-witness pair if and only if there exists a common generalized substring $s$ between $A$ and $B$ such that $|\ts|=\td$, $|s|<n^{1/3}$ and $s$ starts from the beginning of $A[i_A]$ or from the beginning of $B[i_B]$, depending on which one is shorter.
	\begin{proof}\hfill
  	  \begin{itemize}
  		  \item $(\Leftarrow)$:
  			  This is true by \autoref{dfn:witness-pair}.
  		  \item $(\Rightarrow)$:
  			  From \autoref{dfn:witness-pair}, there exists a common substring $t$ that starts within $A[i_A]$ and $B[i_B]$, which implies $C(A[i_A]) = C(t[1]) = C(B[i_B])$.
  			  Assuming $R(A[i_A])\geq R(B[i_B])$, then $t$ can be extended to its left, until $R(t[1]) = R(B[i_B])$, getting another RLE string $t'$.
  			  Because $|\Tt'|\geq|\Tt|\geq\td$ and $|t'|=|t|<n^{1/3}$, the prefix of $t'$ with decoded length $\td$ is the $s$.
  			  A similar argument takes care of the case that $R(A[i_A])\leq R(B[i_B])$.
  	  \end{itemize}
	\end{proof}
\end{pro}

Given a $\td$-witness pair $(i_A,i_B)$, WLOG, we can assume $R(A[i_A])\geq R(B[i_B])$.
From \autoref{thm:shift}, the associated common generalized substring $s$ starts from the beginning of the run $B[i_B]$ and ends within the run $B[j_B]$ where $j_B = \Piv_B(P_B[i_B]+\td-R(B[i_B]))$ since $|\ts| = \td$.
In other words, $s$ can be identified by the 3-tuple $(i_A, i_B, j_B-i_B+1)$.
And since one call to the inverse prefix-sum (to calculate $j_B$) costs $\tO(1)$ time (\autoref{thm:inv-prefix}), finding the 3-tuple can be reduced to finding a $\td$-witness pair $(i_A,i_B)\in[n]\times[n]$, which can be solved with an algorithm similar to the element distinctness algorithm in \autocite{LeGallS22}.

In the setting of \autocite{LeGallS22}, the element distinctness algorithm takes two input lists $X$ and $Y$ (both have size $\Od(n)$), and finds a pair $(x, y)\in X\times Y$ such that $f(x,y)=1$ for some predicate function $f$.
The overall time cost is $\tO(n^{2/3}\cdot T)$ where $T$ is the time cost to evaluate the predicate function $f$.

In our setting, $X$ and $Y$ are indices of $A$ and $B$, \ie $[n]$.
For $x\in X$ and $y\in Y$, $f(x,y) = 1$ if and only if $(x,y)$ is a $\td$-witness pair, which can be verified in time $T=\tO(\sqrt{n^{1/3}})=\tO(n^{1/6})$.
Thus, the time complexity is $\tO(n^{2/3}\cdot n^{1/6})=\tO(n^{5/6})$.
The more in-depth discussion on how to verify a $\td$-witness pair is in \autoref{sec:QW}.
The following lemma summarizes the algorithm for short answers.

\begin{thm}[Algorithm for short answers]\label{thm:small-d}
	Given oracle access to RLE strings $A$ and $B$, both have encoded length $\Od(n)$, and their prefix-sums, with a positive integer $\td \leq \tn$, there exists a quantum algorithm $\calA$, with high probability, finds a 3-tuple $(i_A,i_B,|s|)$ that identifies a common generalized substring (see \autoref{dfn:rle-substr}) $s$ between $A$ and $B$ of the following properties if it exists: $s$ starts within runs $A[i_A]$ and $B[i_B]$ in $A$ and $B$, respectively. $s$ has encoded length $|s| <n^{1/3}$ and decoded length $|\ts|\geq\td$. Otherwise $\calA$ rejects.
  	The algorithm costs $\tO(n^{5/6})$ time.
\end{thm}

\subsubsection{Quantum walk search}\label{sec:QW}

To find a $\td$-witness pair in $[n]\times[n]$, we use the quantum walk framework of \autocite{MagniezNRS11} in the formulation of \autocite{LeGallS22}, in which we walk on a direct product of two Johnson graphs.

A Johnson graph, denoted as $J(n,r)$, consists of $\binom{n}{r}$ vertices, each being an $r$-sized subset $R$ of a list $S$ of size $n$.
In $J(n,r)$, two vertices $R_1$ and $R_2$ are connected iff $|R_1\cap R_2|=r-1$.

In the direct product of two Johnson graphs $J_A(n,r)$ and $J_B(n,r)$, each vertex is a pair of subsets $(R_A,R_B)$, in which $R_A$ comes from $J_A$ and $R_B$ comes from $J_B$, and the vertex is connected to another one $(R'_A,R'_B)$ if and only if $|R_A\cap R'_A|=r-1$ and $|R_B\cap R'_B|=r-1$.
In a search problem, vertices that one wants to find are called \emph{marked} vertices.

Associated with each vertex $(R_A,R_B)$, is a data structure $D(R_A,R_B)$ that supports three operations: setup, update, and checking; whose costs are denoted by $s(r)$, $u(r)$, and $c(r)$, respectively.
The setup operation initializes the $D(R_A,R_B)$ for any vertex $(R_A,R_B)$;
the update operation transforms $D(R_A,R_B)$ into $D(R'_A,R'_B)$ which is associated with a neighbouring vertex $(R'_A,R'_B)$ of $(R_A,R_B)$ in the graph;
and the checking operation checks whether the vertex $(R_A, R_B)$ is marked.
The MNRS quantum walk search algorithm can be summarized as:

\begin{thm}[MNRS Quantum Walk search \autocite{MagniezNRS11}, in the formulation of \autocite{LeGallS22}]\label{thm:MNRS}
Assume the fraction of the marked vertices is zero or at least $\dta$.
Then there is a quantum algorithm that always rejects when no marked vertex exists; otherwise, with high probability, it finds a marked vertex $(R_A,R_B)$.
The algorithm has complexity
	\[\label{eqn:MNRS}
  	  \tO\left(s(r) + \frac{1}{\sqrt{\dta}}\left(\sqrt{r}\cdot u(r)+c(r)\right)\right)
	.\]
\end{thm}

In our setting, the list $S$ is $[n]$, and the Johnson graphs $J_A(n,r)$ and $J_B(n,r)$ are associated with the input RLE strings $A$ and $B$, respectively.
A vertex $(R_A,R_B)$ is marked if and only if a $\td$-witness pair $(i_A,i_B)\in R_A\times R_B$ exists.
The fraction of marked vertices $\dta$ is lower-bounded by $\Omg(r^2/n^2)$ since in the worst case there is only one $\td$-witness pair exists, and thus $\binom{n-1}{r-1}^2$ out of $\binom{n}{r}^2$ pairs of subsets are marked.
The data structure $D(R_A,R_B)$ consists of two instances of another data structure $D_A$ and $D_B$, one for each subset. So $D(R_A,R_B)$ can be viewed as the pair $(D_A(R_A),D_B(R_B))$ with some extra operations on top of it.

\begin{thm}\label{thm:api}
	Associated with an RLE string $A$ of length $n$ and an $r$-sized subset $R_A$ of $[n]$, there is a history-independent data structure $D_A(R_A)$ that supports the following operations:
	\begin{enumerate}[(1)]
  	  \item \textbf{setup}: in $\tO(r\cdot T(n))$ time, setup $D_A(R_A)$ for an $r$-sized subset $R_A$ of $[n]$;
  	  \item \textbf{update}: in $\tO(T(n))$ time, update $D_A(R_A)$ to $D_A(R'_A)$ where $|R'_A\cap R_A|=r-1$;
  	  \item \textbf{checking}: given an index $i_B\in[n]$ associated with another RLE string $B$ and a positive number $\td \leq \tn$, in $\tO(T(n))$ time, with high probability, determine if an $i_A\in R_A$ exists such that $R(B[i_B])\leq R(A[i_A])$ and $(i_A,i_B)$ is a $\td$-witness pair, if so, $i_A$ is returned.
	\end{enumerate}
    Here, $T(n)=n^{1/6}$ is the time cost of comparing any $i_A\in R_A$ to any $i_B\in R_B$.
\end{thm}

Let's see how to build $D(R_A,R_B)$ with $D_A(R_A)$ and $D_B(R_B)$ first and defer the proof of \autoref{thm:api} to \autoref{sec:ds-impl}.

\paragraph{The setup and update operations}
The setup and the update operations can be passed down to $D_A(R_A)$ and $D_B(R_B)$ directly, so from \autoref{thm:api} the setup time $s(r)$ is $\tO(r\cdot T(n))$ and the update time $u(r)$ is $\tO(T(n))$.

\paragraph{The checking operation}
The checking operation of $D(R_A,R_B)$ verifies if there exists a $\td$-witness pair $(i_A,i_B)$ for $i_A\in R_A$ and $i_B\in R_B$.
To build the checking operation for $D(R_A,R_B)$, we do two Grover's searches. %
In the first Grover's search, we search over $i_B\in R_B$ and call the checking operation of $D_A(R_A)$ to see whether there exists an $i_A\in R_A$ such that $(i_A,i_B)$ is a $\td$-witness pair and $R(B[i_B])\leq R(A[i_A])$ (or it verifies the absence of such an index).
In the other Grover's search, the roles of $A$ and $B$ are swapped to cover the case of $R(B[i_B]) \geq R(A[i_A])$.
By \autoref{thm:api}, the checking time for both $D_A(R_A)$ and $D_B(R_B)$ is $\tO(T(n))$, so the checking time for $D(R_A, R_B)$ is $c(r)=\tO(\sqrt{r}\cdot T(n))$.

Plugging in $\dta=\Omg(r^2/n^2)$, $s(r)=\tO(r\cdot T(n))$, $u(r)=\tO(T(n))$, and $c(r)=\tO(\sqrt{r}\cdot T(n))$ into \autoref{eqn:MNRS}, and multiplied with the comparison cost $T(n) = \tO(n^{1/6})$, the time cost of the quantum walk search is
\[
	\tO\left(r + \frac{n}{r}\left(\sqrt{r}+\sqrt{r}\right)\right) \cdot \tO\left(n^{1/6}\right).
\]
With $r=n^{2/3}$, the time cost is minimized at $\tO(n^{5/6})$.

After getting a vertex $(R_A,R_B)$ that contains a $\td$-witness pair from the quantum walk search algorithm (\autoref{thm:MNRS}), we do one extra checking operation on $D(R_A,R_B)$ to find the $\td$-witness pair $(i_A,i_B)$, from which, with inverse prefix-sum, we can find the ending runs of the common substrings $s$, and further calculate the encoded length of $s$.
This proves \autoref{thm:small-d}.

\subsubsection{Proof of \autoref{thm:api} (Data structure)}\label{sec:ds-impl}

Skip list \autocite{Pugh90} is a history-independent probabilistic data structure.
In Ambainis' Element Distinctness algorithm \autocite{Ambainis07}, skip lists are used to maintain a sorted array that supports insertion and deletion.
In \autocite{BuhrmanLPS22}, Buhrman, Loff, Patro and Speelman augmented Ambainis' skip list with indexing operation.
Recently, Akmal and Jin extended the data structure further to support range-maximum query \autocite{AkmalJ22}.
The capabilities of a skip list are summarized as below.

\begin{lma}[skip list, a mix of \autocite{Ambainis07}, \autocite{BuhrmanLPS22} and \autocite{AkmalJ22}]\label{thm:skip-list}
	There is a history-independent data structure that maintains an array of $\Od(r)$ elements $e_1,e_2,e_3,\ldots,e_r$ that are sorted with regard to some function $f$.
	The data structure supports the following operations, and all of them succeed with high probability\footnote{Each operation is aborted after its denoted running time. This might break the data structure, but Ambainis showed that this doesn't impact the quantum walk algorithm too much. The detailed analysis can be found in \autocite[Lemma 5 and 6]{Ambainis07}.}.
	\begin{enumerate}
  	  \item \textbf{Insertion}:
  		  Given a new element, insert it into the array.
  	  \item \textbf{Deletion}:
  		  Given an element in the array, remove it from the array.
	\end{enumerate}
	The operations above both cost $\tO(T_f)$ time where $T_f$ is the time to determine whether $f(e_i)\leq f(e_j)$ for $i$, $j$ in $[r]$.
	\begin{enumerate}[resume]
  	  \item \textbf{Indexing}:
  		  Given an index $1\leq i \leq r$, return the $i$-th element in the array.
  		  This operation costs $\tO(1)$ time.
  	  \item \textbf{Range-maximum query}:
  		  Given two indices $1 \leq i \leq j \leq r$, return $k\in[i,j]$ such that $g(e_k)=\max_{i\leq l \leq j}\{g(e_l)\}$ where $g$ is some function.
  		  This operation costs $\tO(T_g)$ time where $T_g$ is the time to determine whether $g(e_a)\leq g(e_b)$ for $a$, $b$ in $[r]$.
	\end{enumerate}
\end{lma}

We use a skip list to build the data structure $D_A(R_A)$ introduced in \autoref{thm:api}.
It stores an $\Od(r)$-sized subset $R_A$ of $[n]$ associated with an RLE string $A$ of length $\Od(n)$.
In the following text, we show what functions $f$ and $g$ are used.
After that, we show how to support the operations of $D_A(R_A)$ listed in \autoref{thm:api} with the skip list.

\paragraph{The function $g$ for range-maximum query}
Here, the function $g$ maps an index $i_A$ of the associated RLE string $A$ to the length of the run $A[i_A]$, \ie $g(A, i_A) = R(A[i_A])$.
Since this uses $\Od(1)$ query to the input string $A$ and $R(A[i_A])$ is a number, the time cost $T_g$ is $\Od(1)$.

\paragraph{The function $f$ for comparing elements}
We define $f$ to be a function that maps an index $i_A\in[n]$ of the associated RLE string $A$ to a tuple:
\[\label{eqn:cmpf}
	f: i_A \mapsto \left(v(A, i_A, n^{1/3}),\ g(A, i_A),\ i_A\right)
\]
where $g$ is the function for the range-maximum query, and for some RLE string $S$, $v(S,i,\ell)$ denotes a substring starting from the run $S[i]$ of encoded length at most $\ell$ with the length of its first run set to $1$.
More precisely,
\[\label{eqn:val-prefix}
	v: (S,i, \ell) \mapsto C(S[i])^1S[i+1:i+\ell-1].
\]
For example, $v(\torle{aabbbbbbccccddeeee}, 2, 3)$ is $\torle{bccccdd}$.

\begin{figure}[H]
	\centering
  \ifx\included\undefined
\documentclass[crop]{standalone}
\usepackage{tikz}
\usetikzlibrary{
	decorations.pathmorphing,
	calc,
	matrix,
}
\begin{document}
\input{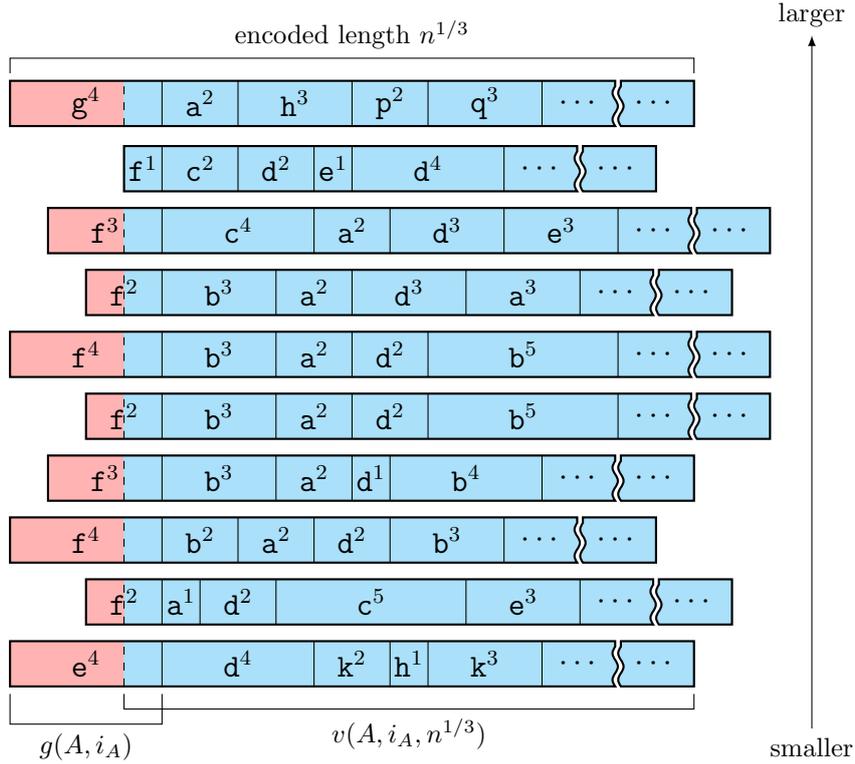}
\fi

\begin{tikzpicture}

\newcommand{\h}{0.6}
\newcommand{\xs}{0.5}

\matrix[row sep=.5em, outer sep=1em] (M) {
	\drawrle{g4a2h3p2q3}
	\coordinate(a) at ($(R)+(0,\h+0.1)$);
	\coordinate(b) at ($(B)+(0,\h+0.1)$);\\

	\drawrle{f1c2d2e1d4}\\

	\drawrle{f3c4a2d3e3}\\

	\drawrle{f2b3a2d3a3}\\

	\drawrle{f4b3a2d2b5}\\

	\drawrle{f2b3a2d2b5}\\

	\drawrle{f3b3a2d1b4}\\

	\drawrle{f4b2a2d2b3}\\

	\drawrle{f2a1d2c5e3}\\

	\drawrle{e4d4k2h1k3}
	\coordinate(c) at ($(R)-(0,0.1)$);
	\coordinate(d) at (-\xs,-0.1);
	\coordinate(e) at ($(B)-(0,0.1)$);\\
};

\draw[latex-]
	(M.north east) node[above, desc]{larger}
	--
	(M.south east) node[below, desc]{smaller};

\begin{scope}[every node/.style={brkt, desc}]
\draw (a)          --++(0,   \h/3) -| node [above] {encoded length $n^{1/3}$} (b);
\draw (d)          --++(0,  -\h/3) -| node [below] {$v(A,i_A,n^{1/3})$} (e);
\draw (d)++(\xs,0) --++(0,-2*\h/3) -| node [below] {$g(A,i_A)$} (c);
\end{scope}

\end{tikzpicture}

\ifx\included\undefined
\end{document}
\fi
	\caption{
    	An illustration of how $i_A$ in the data structure $D_A(R_A)$ are sorted.
    	The order of $i_A$ in $D_A(R_A)$ is determined by comparing their corresponding substring $A[i_A:i_A+n^{1/3}-1]$ (the rows in the figure).
    	More precisely, the order of two rows is determined by comparing first the blue parts (\autoref{eqn:val-prefix}) then the red parts (essentially, the length of the first runs).
    	After that, the indices themselves are used to break ties.
	}
	\label{fig:order}
\end{figure}

For any $i_A,i'_A\in[n]$, we compare $f(i_A)$ and $f(i'_A)$ entry-to-entry from left to right.
We say $v(A,i_A, n^{1/3}) \leq v(A,i'_A, n^{1/3})$ if the former is lexicographically smaller or equal to the latter when they are decoded.
More precisely, with the former and the latter denoted as $s_A$ and $s'_A$, respectively, we do the comparison in the following way:
\begin{enumerate}[(1)]
	\item \label{itm:cmpf-1}
  	  We use minimal finding to find the smallest $k$ such that $s_A[k]\neq s'_A[k]$.
  	  If no such $k$ exists, we determine the order according to $|s_A|$ and $|s'_A|$.
  	  Otherwise, we proceed to the next step.
	\item If $C(s_A[k]) \neq C(s'_A[k])$, we use them to determine the order.
  	  Otherwise, it must be the case that $R(s_A[k]) \neq R(s'_A[k])$, and we proceed to the next step.
	\item WLOG, let's assume $R(s_A[k]) < R(s'_A[k])$.
  	  Then if $k=|s_A|$, $i_A$ comes before $i'_A$ since it is shorter.
  	  Otherwise, we determine the order based on $C(s_A[k+1])$ and $C(s'_A[k])$.
\end{enumerate}

As a result, the time cost for checking whether $f(i_A)\leq f(i'_A)$ is dominated by the minimal finding in step \ref{itm:cmpf-1}, and thus with the guarantee that both $|v(A,i_A,n^{1/3})|$ and $|v(A,i'_A,n^{1/3})|$ are in $\Od(n^{1/3})$, the time cost for one comparison $T_f$ is $\tO(\sqrt{n^{1/3}})=\tO(n^{1/6})$.

Now, let's see how to build the three operations -- setup, update, and checking -- in \autoref{thm:api} from the ones provided by the data structure in \autoref{thm:skip-list}.

\paragraph{The setup and update operations}
To set up $D_A(R_A)$, we insert the elements of $R_A$ into an empty skip list one-by-one, which costs $\tO(rn^{1/6})$ time.
To update $D_A(R_A)$ to $D_A(R'_A)$, the only element in $R_A \setminus R'_A$ is deleted from the skip list, then the one in $R'_A\setminus R_A$ is inserted, getting time complexity $\tO(n^{1/6})$.

\paragraph{The checking operation}
The checking operation of $D_A(R_A)$ takes an index $i_B$ associated with another RLE string $B$ and a positive integer $\td\leq\tn$ as inputs.
Its output is an index $i_A\in R_A$ such that $(i_A,i_B)$ is a $\td$-witness pair and $R(B[i_B])\leq R(A[i_A])$.

\begin{col}\label{thm:shift-skew}
	For $(i_A,i_B)\in[n]\times[n]$, if $R(A[i_A])\geq R(B[i_B])$ and $(i_A,i_B)$ is a $\td$-witness pair, then $s$ is the associated common generalized substring if and only if
	\begin{enumerate}[(1)]
  	  \item $s$ is a generalized substring of $B$ that starts from the beginning of $B[i_B]$,
  	  \item $|\ts|=\td$,
  	  \item $|s|<n^{1/3}$
  	  \item $s$ is a generalized substring of $A$ that starts within the run $A[i_A]$.
	\end{enumerate}
	\begin{proof}\hfill
  	  This is the direct result of combining $R(A[i_A]) \geq R(B[i_B])$ with \autoref{thm:shift}.
	\end{proof}
\end{col}

For $i_A\in R_A$ and $i_B\in R_B$, due to (1) and (2) in \autoref{thm:shift-skew}, there is only one possible $s$ we have to check -- the one that starts from the beginning of $B[i_B]$ and has a decoded length $\td$.
The $s$ ends within the run $B[j_B]$ where $j_B = \Piv_B(\tj_B)$ with $\tj_B=P_B[i_B]+\td-R(B[i_B])$.
More precisely, the $s$ ends at the $r_B$-th character in the run where $r_B = \tj_B-P_B[j_B-1]$, and thus it is in the following form
\begin{equation}\label{eqn:small-s}
	s := B[i_B:j_B-1]C(B[j_B])^{r_B}.
\end{equation}
If $j_B-i_B+1>n^{1/3}$ (the encoded length is too long) or $P_B[i_B]+\td-B[i_B]>\tn$ (the decoded length is too short), the $s$ does not meet (1), (2), and (3) at the same time so the checking operation rejects.\footnotemark{}
Then we check if $s$ fulfills (4), which can be further split into following parts:
\begin{enumerate}[(a)]
	\item $C(A[i_A]) = C(s[1])$ and $s[2:|s|-1]=A[i_A+1:i_A+|s|-2]$.
	\item $C(A[i_A+|s|-1])=C(s[|s|])$ and $R(A[i_A+|s|-1])\geq R(s[|s|])$
	\item $R(A[i_A])\geq R(s[1])$.
\end{enumerate}

\footnotetext{Since conditions (1), (2), and (3) are independent to $A$, it's possible to check them in the checking operation of $D(R_A,R_B)$, just before calling the checking operation of $D_A(R_A)$.}

With the $v$ defined in \autoref{eqn:val-prefix}, (a) and (b) can be rewritten as the decoding of $v(s,1, |s|)$ being a prefix of the decoding of $v(A, i_A, |s|)$.
Because $|s|<n^{1/3}$ and the elements in the skip list are first sorted lexicographically \wrt $v(A, i_A, n^{1/3})$, of which $v(A, i_A, |s|)$ is a prefix, all $i_A\in R_A$ that meet (a) form a interval in the skip list, whose starting point and ending point can be located with binary search.
After that, we use range-maximum query to find a $k_A$ in the interval such that $R(A[k_A])\geq R(A[i_A])$ for all $i_A$ in the interval.%
\footnote{$R(A[i_A])$ is not necessarily sorted in the interval because it's possible that for some $i'_A,i_A\in R_A$ that $R(A[i'_A]) > R(A[i_A])$ and both fulfill (a), but $i'_A$ comes before $i_A$ in the skip list due to $v(A,i'_A,\ell')$ being lexicographically smaller than $v(A,i_A,\ell)$. See \autoref{fig:checking} for an example.}
If $R(A[k_A])\geq R(s[1])$, $s$ meets (c), and thus $(k_A,i_B)$ is a $\td$-witness pair so $(k_A,i_B)$ is returned.
Otherwise, if $R(s[1]) > R(A[k_A]) \geq R(A[i_A])$ for all $i_A\in R_A$, no $i_A\in R_A$ is able to form a $\td$-witness pair with the given $i_B$ so the checking operation rejects.

\begin{figure}[H]
	\centering
    \ifx\included\undefined
\documentclass[crop]{standalone}
\usepackage{tikz}
\usetikzlibrary{
	decorations.pathmorphing,
	calc,
	matrix,
}
\begin{document}
\input{../pic}
\fi

\begin{tikzpicture}
\newcommand{\h}{0.6}
\newcommand{\xs}{0.5}

\matrix[row sep=.5em]{
	\drawrle{z3f3b3a2d4x3}[2][][9]
	\node at (HEAD)[left, outer sep=.5em] {$B$};
	\draw (R)++(0,\h+0.1)
	      -- ++(0,\h/3)
	      -| node [brkt,above]{$s$}
	      ($(B)+(0,\h+0.1)$);\\

	\path[use as bounding box](0,0) -- (0,\h);\\

	\drawrle{g4a2h3p2q3}[1][4]\\

	\drawrle{f1c2d2e1d6}[1][4]\\

	\drawrle{f1b3a2d3a3}[1][4]
	\coordinate (a) at ($(HEAD)+(-0.4,\h/2)$);\\

	\drawrle{f4b3a2d2b5}[1][4]\\

	\drawrle{f2b3a2d2b5}[1][4]\\

	\drawrle{f3b3a2d1b4}[1][4]
	\coordinate (b) at ($(HEAD)-( 0.4,\h/2)$);\\

	\drawrle{f4b2a2d2b3}[1][4]\\

	\drawrle{f2a1d2c5e3}[1][4]\\

	\drawrle{e4d4k2h1k3}[1][4]
	\draw (-\xs, -0.1)
	      -- ++(0, -\h/3)
	      -| node [below, brkt, desc]
	              {$v(A,i_A,|s|)$}
	      ($(B)-(0,0.1)$);
	\draw (0, -0.1)
	      -- ++(0,-2*\h/3)
	      -| node [below, brkt, desc]
	              {$g(A,i_A)$}
	      ($(R)-(0,0.1)$);\\
};

\draw (a)
      -- ++(-2.2,0)
      |- node[brkt, desc, left, text width=2cm, inner sep=0pt]
             {The interval for $s$}
      (b);
\end{tikzpicture}

\ifx\included\undefined
\end{document}
\fi
	\caption{
    	An example of finding an $i_A\in R_A$ that meets requirement (4) in \autoref{thm:shift-skew} for $i_B=2$, $B=\torle{zzzfffbbbaaddddxxx}$, and $\td=9$.
    	The top-most row shows the corresponding $s=\torle{fffbbbaad}$ that we are going to check (\autoref{eqn:small-s}).
    	Other rows are some entries in $D_A(R_A)$ (see \autoref{fig:order}).
    	In each row, the blue part is $v(A,i_A,|s|)$, and the decoded length of the red one is $g(A, i_A)-1$.
    	The line on the left marks the interval, in which the blue parts all start with $v(B,i_B,|s|)$.
    	Range-maximal query then finds a $k_A$ that maximizes $g(A,k_A)$ in the interval.
    	Since  $g(A,k_A)=4\geq 3=g(B,i_B)$, there exists a $\td$-witness pair $(k_A,i_B)$.
	}
	\label{fig:checking}
\end{figure}

In terms of time cost for the checking operation, identifying the $s$ that fulfills (1), (2), and (3) costs $\tO(1)$ time; the binary searches that locate the interval compare $v(s,1,|s|)$ with $v(A,i_A,|s|)$ for at most $\Od(\log r)$ different $i_A\in R_A$, using $\Od\left(\log(n^{2/3})\right) \cdot\tO\left(\sqrt{|s|}\right)=\tO(n^{1/6})$ time; and the time cost for the range-maximal query is $\tO(1)$. Hence one checking operation costs $\tO(n^{1/6})$ time. This completes the proof of \autoref{thm:api} and wraps up the algorithm for small answers.

\subsection{Algorithm for long answers}\label{sec:large-d}

In this section, we show how to find a common generalized substring $s$ of decoded length $\td$ with encoded length at least $n^{1/3}$, which will be referred to as a ``long answer" in the following text.

For a given $\td$, we search over the all possible encoded length, $[n^{1/3}, n]$, to find a long answer. We will introduce a procedure that, given an input $\ell\in[n]$, it finds a long answer of encoded length between $\ell$ and $2\ell$ in $\tO\PP*{n/\sqrt{\ell}}$ quantum time  (\autoref{thm:sample-pair}). This procedure is called repeatedly for $\ell=n^{1/3}$, $2n^{1/3}$, $4n^{1/3}, 8n^{1/3}$, and so on, until a long answer is found or until $\ell$ reaches $2^kn^{1/3}>n$ for some integer $k$. Because the search range doubles between calls, the procedure is executed for at most $k=\ceil{\log_2(n/n^{1/3})}\in\Od(\log n)$ times. Hence the overall time cost is
\[
	\sum_{k=0}^{\Od(\log n)} \tO\PP*{\frac{n}{\sqrt{2^kn^{1/3}}}}
	< \sum_{k=0}^{\Od(\log n)} \tO\PP*{\frac{n}{\sqrt{n^{1/3}}}}
	\in \tO\PP*{n^{5/6}},
\]
which coincides with the time complexity of the algorithm for short answers (\autoref{thm:small-d}).
The following theorem summarizes the algorithm for long answers.

\begin{thm}[Algorithm for long answers]\label{thm:large-d}
	Let $A$ and $B$ be RLE strings of encoded length $n$ and decoded length $\tn$.
	Given oracle access to $A$, $B$, and their prefix-sums, for all $\td\in[\tn]$, there exists a quantum algorithm $\cA$ that finds a tuple of integers $(i_A,i_B, |s|)$ which identifies a common generalized substring (see \autoref{dfn:rle-substr}) $s$ between $A$ and $B$ of the following properties if it exists:
	\begin{itemize}
  	  \item $s$ starts within runs $A[i_A]$ and $B[i_B]$ in $A$ and $B$, respectively;
  	  \item The encoded length of $s$ is $|s|\geq n^{1/3}$; and
  	  \item The decoded length of $s$ is $|\ts|\geq \td$.
	\end{itemize}
	Otherwise, $\cA$ rejects.
	The time cost of $\cA$ is $\tO(n^{5/6})$ and it succeeds with high probability.
\end{thm}

\subsubsection{Good pairs}

\begin{dfn}[Good pair (modified from \autocite{LeGallS22})]\label{dfn:good-pair}
	Let $\td$ and $\ell$ be integers.
	For two RLE strings $A$ and $B$ of encoded length $n$, we say the pair $(P,S)$, where $P:=A[i_P:i_P+2\floor{\frac13 \ell}-1]$ and $S:=B[i_S:i_S+\ell-1]$, is a \emph{good pair} if and only if
	\begin{enumerate}[(1)]
  	  \item There exists a ``long answer" $s$ between $A$ and $B$, such that $|\ts|\geq \td$ and $|s|\geq \ell$.
  	  \item $s$ starts within runs $A[i_A]$ and $B[i_B]$ in $A$ and $B$, respectively.
        	\Ie there exists a pair $(i_A,i_B)\in[n]\times[n]$ such that $s$ is a common generalized substring between $A[i_A:i_A+|s|-1]$ and $B[i_B:i_B+|s|-1]$.
  	  \item In $A$, $P$ resides within the ``middle part" of $s$ in $A$, \ie $A[i_A+1:i_A+|s|-2]$. More precisely, $i_A+1\leq i_P$ and $i_P+2\floor{\ell/3}-1\leq i_A+|s|-2$.
  	  \item Given (3), in $B$, $S$ covers the corresponding $P$ in $B$.
  		  More precisely, $i_S\leq i_B+(i_P-i_A)$ and $i_B+(i_P-i_A)+2\floor{\ell/3}-1\leq i_S+\ell-1$.
	\end{enumerate}
\end{dfn}

By (1), if no long answer exists or one exists but $|s|<\ell$, no good pair exists.
Suppose at least one long answer $s$ exists and $|s|\geq\ell$. When a pair of substrings $(P,S)$ is sampled uniformly\footnotemark, the probability that $P$ meets (3) is $\Omg\PP*{\frac{|s|-2-|P|}{n}}=\Omg\PP*{\frac{\ell-|P|}{n}}=\Omg\PP*{\frac{\ell}{n}}$; given (3), the probability of $S$ meeting (4) is $\Omg\PP*{\frac{\ell-|P|}{n}}=\Omg\PP*{\frac{\ell}{n}}$.
As a result, with probability $\Omg\PP*{\frac{\ell^2}{n^2}}$, $(P,S)$ is a good pair.
\footnotetext{More precisely, we sample $(i_P,i_S)$ uniformly from $[n-2\floor{\ell/3}]\times[n-\ell]$.}

\begin{figure}[H]
	\centering
    \ifx\included\undefined
\documentclass[tikz,crop]{standalone}
\usetikzlibrary{
	calc,
	matrix,
}
\begin{document}
\input{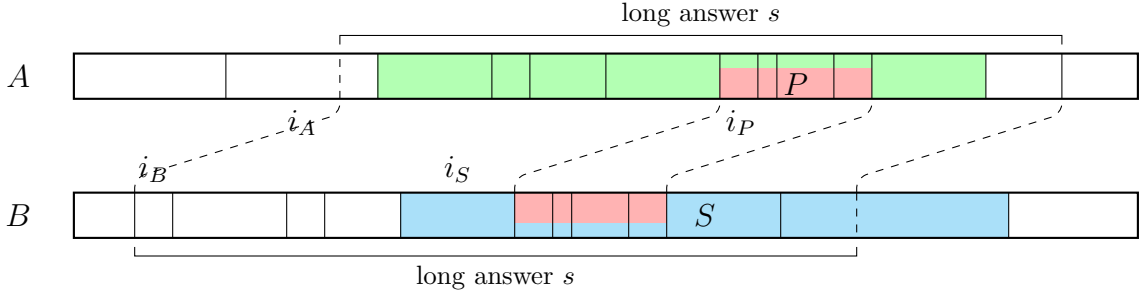}
\fi

\begin{tikzpicture}
	\newcommand{\h}{0.6}
	\newcommand{\As}{5.0}
	\newcommand{\Bs}{2.3}
	\matrix [row sep=0em, column sep=1em] {
		\node {$A$};
		&
		\begin{scope}[yscale=\h,shift={(\As,-1/2)}]
			\fill[fill=green!30] (-1,0) rectangle (7,1);
			\fill[fill=red!30]  (3.5,0) rectangle node{$P$} (5.5,2/3);
			\node at (3.5/2+4.25/2-0.1,0) [below]  {$i_P$};
			\node at (-3/2+-1/2,0)        [below]  {$i_A$};
			\draw (-1.5,1.2) --++(0,0.2) -| node[brkt, desc, above]{long answer $s$} (8,1.2);
			\foreach \bd in {-3, -1, 0.5, 1, 2, 3.5, 4, 4.25, 5, 5.5, 7, 8}
				\draw (\bd,0) -- ++(0, 1);
			\draw[thick] (-\As,0) rectangle ++(14, 1);
			\coordinate (A) at (0,0);
		\end{scope}
		\\
		\node {$B$};
		&
		\begin{scope}[yscale=\h,shift={(\Bs,-1/2)}]
			\fill[fill=cyan!30] (2,0) rectangle node{$S$} (10,1);
			\fill[fill=red!30] (3.5,1/3) rectangle (5.5,1);
			\node at (-1.5/2+-1/2,1) [above]  {$i_B$};
			\node at (2/2+3.5/2,1)   [above]  {$i_S$};
			\draw (8,-0.2) --++(0,-0.2) -| node[brkt, desc, below]{long answer $s$} (-1.5,-0.2);
			\foreach \bd in {-1.5, -1, 0.5, 1, 2, 3.5, 4, 4.25, 5, 5.5, 7, 10}
				\draw (\bd,0) -- ++(0, 1);
			\draw[thick] (-\Bs,0) rectangle ++(14, 1);
			\coordinate (B) at (0,1);
		\end{scope}
		\\
	};
	\begin{scope}[yscale=\h, every path/.style={dashed}]
		\foreach \i in {-1.5, 3.5, 5.5, 8.0}
			\draw[rounded corners]
				(B)++(\i,0) -- ++(0,1/3) -- ($(A)+(\i,-1/3)$) -- ++(0,1/3);
		\draw (A)++(-1.5,0) -- ++(0, 1)
		      (B)++(+8.0,0) -- ++(0,-1);
	\end{scope}
\end{tikzpicture}

\ifx\included\undefined
\end{document}
\fi%
	\caption{An example of a good pair $(P,S)$. The solid vertical lines represent the boundaries of the runs. The red part in $S$ marks the corresponding $P$ in $B$. The indices are for RLE strings. Note that $s$ covers only a fraction of $A[i_A]$, which is why we need $P$ resides within the ``middle part" of $A[i_A:i_A+|s|-1]$ (the green part).}
	\label{fig:enter-label}
\end{figure}

From a good pair $(P,S)$, by \autoref{thm:grow-pair}, one can find a long answer $s$ such that $\ell\leq|s|\leq2\ell$ (if it exists) in $\tO(\sqrt{\ell})$ quantum time.
Therefore, applying \autoref{thm:grow-pair} on a randomly sampled pair $(P,S)$ gives a procedure to find a long answer of encoded length between $\ell$ and $2\ell$ with one-sided error rate $1-\Omg\PP*{\frac{\ell^2}{n^2}}$. Using amplitude amplification \autocite{BrassardH97} and \autoref{thm:whp}, the success probability can be boosted to $\Omg(1-1/\poly(n))$ with the overall time cost being
\[
	\tO\PP*{\sqrt{\frac{n^2}{\ell^2}}\cdot\sqrt{\ell}}
	= \tO\PP*{\frac{n}{\sqrt{\ell}}}.
\]
The following lemma summarizes the procedure.

\begin{lma}\label{thm:sample-pair}
	For all $\ell\in[n]$, if a long answer $s$ such that $|s|\in[\ell,2\ell]$ exists, it can be found in $\tO\PP*{\frac{n}{\sqrt{\ell}}}$ time, with high probability; Otherwise, its absence can be confirmed.
\end{lma}

\subsubsection{Search for a long answer from a good pair}

In this section, we show how to find a ``long answer" from a good pair $(P,S)$. By definition, there is a long answer $s$ of decoded length $|\ts|\geq\td$ exists and $|s|\geq|S|$.

\begin{lma}\label{thm:grow-pair}
Given a good pair $(P, S)$, with high probability, a long answer of encoded length between $|S|$ and $2|S|$ can be found in $\tO(\sqrt{|S|})$ quantum time. Or its absence can be confirmed.
\end{lma}

The first and the last instance of $P$ in $S$ are located by applying the string matching algorithm from \autocite{HariharanV03} with each run -- a character-length pair -- being treated as an atomic element. Both instances can be found in $\tO(\sqrt{|S|}+\sqrt{|P|})=\tO(\sqrt{|S|})$ quantum time. Let the first and the last instance start from runs $B[l]$ and $B[r]$, respectively.

If $l=r$, there is only one instance of $P$ in $S$, and it corresponds to $P$ (more precisely, $A[i_P:i_P+|P|-1]$). We then find the longest common suffix between (the decodings of) $A[1:i_P-1]$ and $B[1:l-1]$ using \autoref{thm:RLE-lcp}, with its encoded length capped under $|S|$ in $\tO(\sqrt{|S|})$ time. Similarly, we find the longest common prefix of $A[i_P+|P|:n]$ and $B[l+|P|:n]$. Now the suffix, $P$, and the prefix together form a common generalized substring between $A$ and $B$, whose decoded length is then compared to $\td$ to determine if it is a long answer.

In the case that $l\neq r$, \ie there are multiple instances of $P$ in $S$. Since $2|P|=4\floor{|S|/3} > |S|$, the two instances of $P$ in $S$ overlap with each other, and thus $P$ is $(r-l)$-periodic (with each character-length pair being counted as an element). As a consequence, the combined substring $T:=S[l:r+|P|-1]$ is $(r-l)$-periodic as well.

We then extend $P$ from its two ends, with the $(r-l)$-periodicity maintained (with a character-length pair begin counted as an element) and its encoded length capped under $2|S|$. We also extend $T$ in the same way.
Let $A[\afa:\bta]$ and $B[\afa':\bta']$ be the extended results from $P$ and $T$, respectively. Now, $s$ and $A[\afa:\bta]$ must have at least one of the following relation:
\begin{enumerate}[(1)]
	\item $s$ starts before $A[\afa]$.
	\item $s$ ends after $A[\bta]$.
	\item $s$ starts and ends within $A[\afa:\bta]$.
\end{enumerate}
\autoref{fig:periodic-cases} depicts these three cases.

The first case implies that the first runs of the extended parts, \ie $A[\afa]$ and $B[\afa']$, are the first periodic runs in the long answer, thus they correspond to each other. In the second case, a similar argument also works, with $A[\bta]$ corresponds to $B[\bta']$. For these two cases, from the corresponding runs, we follow the same procedure as in the $l=r$ case to find a long answer in $\tO(\sqrt{|S|})$ time.

As for the last case, we want to determine which instance of $P$ in $B[\afa':\bta']$ corresponds to the $P$ we sampled from $A$, which by the definition of good pairs is a substring of a long answer.
Suppose there is an instance of $P$ that starts from the run $B[j_P]$ where $\afa'\leq j_P\leq \bta'-|P|+1$.
If $i_P-\afa\leq j_P-\afa'$, the longest matching that it can be extended into has encoded length
\[
  	  (i_P-\afa) + \min(\bta-i_P,\bta'-j_P)
  	  = \min(\bta-\bta', i_P-j_P) + (\bta'-\afa)
\] so a smaller $j_P$ gives a potentially longer result.
On the other hand, if $i_P-\afa\geq j_P-\afa'$, the longest encoded length is
\[
  	  (j_P-\afa') + \min(\bta-i_P,\bta'-j_P)
  	  = \min(j_P-i_P, \bta'-\bta) + (\bta - \afa')
\] and a larger $j_P$ is preferable.
As a result, extending from two specific instances of $P$ -- the last one that starts before $B[\afa'+(i_P-\afa)+1]$ and the first one that starts after $B[\afa'+(i_P-\afa)-1]$ -- then check if any of them is longer than $\td$ when decoded is enough to rule out other choices of $P$ in $B[\afa':\bta']$. Both instances can be located by minimum finding in $\tO(\sqrt{|S|})$ time, and extending them (in the same way as in the $r=l$ case) costs another $\tO(\sqrt{|S|})$.

To complete the proof of \autoref{thm:grow-pair}, the following lemma shows how to find the longest common prefix and the longest common suffix, which are used to extend two corresponding substrings to construct a long answer.

\begin{lma}\label{thm:RLE-lcp}
    Given two RLE strings $A$ and $B$ of encoded length $n$, the RLE of the longest common prefix of $\tA$ and $\tB$ can be found in $\tO(\sqrt{n})$ quantum time with success probability $\Omg(1-1/\poly(n))$.
	The same statement is also true for the RLE of the longest common suffix.
\begin{proof}
First use minimal finding to find the first runs that $A$ and $B$ differ, \ie find the smallest $i$ such that $A[i] \neq B[i]$. If they differ in their characters, \ie $C(A[i]) \neq C(B[i])$, the substring $A[1:i-1]$ decodes to the longest common prefix. Otherwise, if the characters are the same, then it is $A[1:i-1]C(A[i])^r$, where $r = \min(R(A[i]), R(B[i]))$, that decodes to the longest common prefix.
Finally, we can boost the success probability to $\Omg(1-1/\poly(n))$ with \autoref{thm:whp}.
By reversing the indices (mapping $1$ to $n$, $2$ to $n-1$, and so on), the procedure above finds the longest common suffix.
	\end{proof}
\end{lma}

\begin{figure}[H]
	\centering
	\input{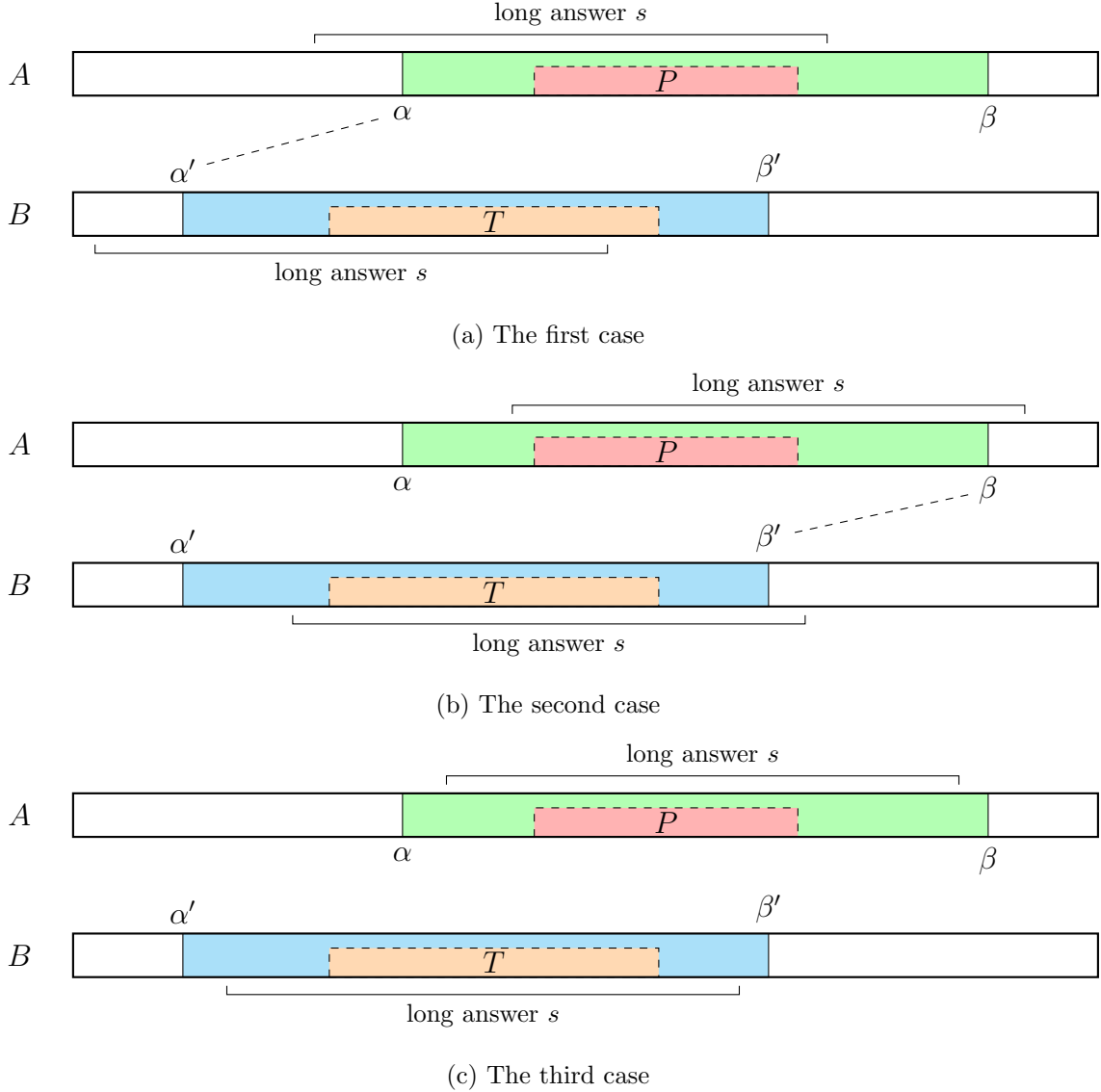}
	\caption{The three possible cases of the long answer $s$ with regard to $A[\afa:\bta]$. The green part ($A[\afa:\bta]$) and the blue part ($B[\afa':\bta']$) are extended from $P$ and $T$, respectively. The dashed lines connect the corresponding indices.}
	\label{fig:periodic-cases}
\end{figure}

\subsection{Putting things together}

To solve $\LCSRLEP$, we do a binary search over $[1,\tn]$. In every iteration, the algorithm for small answers (\autoref{thm:small-d}) and the algorithm for long answers (\autoref{thm:large-d}) are called. Since both sub-algorithm have time cost $\tO(n^{5/6})$ and the binary search ends within $O(\log \tn)$ iterations, the overall time cost is $\tO(n^{5/6})\cdot\Od(\log\tn)$.

\begin{thm}[Algorithm for $\LCSRLEP$]
	Given oracle access to RLE strings $A$ and $B$, and their prefix-sums, there exists a quantum algorithm $\cA$, with high probability, finds a 3-tuple $(i_A,i_B,|s|)$ that identifies a longest common generalized substring (see \autoref{dfn:rle-substr}) $s$ between $A$ and $B$ if it exists; otherwise $\calA$ rejects. $\cA$ has a time cost $\tO(n^{5/6})\cdot\Od(\log\tn)$, where $n$ and $\tn$ are the encoded length and decoded length of input strings, respectively.
\end{thm}

\newcommand{\ta}{\texttt{a}}
\newcommand{\tb}{\texttt{b}}
\newcommand{\tc}{\texttt{c}}
\newcommand{\tat}{\texttt{@}}
\newcommand{\tsr}{\texttt{\#}}
\section{Lower Bounds}\label{sec:LB}

In this section, we investigate the time complexity lower bound for calculating the length (encoded and decoded) of longest common substring from two RLE strings without access to prefix-sum oracles. And use the results to show the lower bound of finding an LCS from two RLE strings is $\tOmg(n)$, which is our motivation and justification to introduce the prefix-sum oracle.

\subsection{Lower Bound on \texorpdfstring{$\DLLCSRLE$}{DLLCSRLE}}\label{sec:LB-DLLCSRLE}

In this section we show how to reduce $\PARITY$ to $\DLLCSRLE$, obtaining the following result.

\begin{lma}[Lower Bound of $\DLLCSRLE$]\label{thm:LB-DLLCSRLE}
Any quantum oracle algorithm $\calA$ requires at least $\Omg(n)$ queries to solve $\DLLCSRLE$, with probability at least $2/3$.
\end{lma}

The main idea is to encode an $n$-bit binary string $B=B_1B_2\ldots B_n$ as an RLE string $S_B$,
in which $R(S_B[i]) = 2+B_i$.
Then, using $\calA$, we find the length of LCS of $S_B$ with itself.
Here, what $\calA$ outputs is basically the decode length of $S_B$, \ie $|\wtilde{S_B}|$.
From that, we can calculate the parity of $B$ easily,
and thus \autoref{thm:LB-DLLCSRLE} is proven.

\begin{proof}
Given an $n$-bit binary string $B$, we can construct an RLE string $S_B$ as:
\begin{equation}\label{eqn:SB}
    S_B
  	  = \rle{a,B_1+2,b,B_2+2,a,B_3+2,b,B_4+2}\cdots\gma^{B_n+2},
\end{equation}
where $\gma$ is $\ta$ if $n$ is odd, otherwise it is $\tb$.

Then we assume the algorithm $\calA$ exists.
With $S_B$ as the inputs, the output of $\calA$, the decoded length of LCS between $S_B$ and itself, is
\[
    \calA(S_B,S_B)
    = |\tS_B|
    = \sum_{i=1}^n (B_i+2)
    = 2n+\sum_{i=1}^n Bi,
\]
which has the same parity as $\bigoplus_{i\in[n]} B_i$, \ie the parity of $B$.
Therefore, by checking the lowest bit of $\calA(S_B,S_B)$, we can solve $\PARITY$ with no extra query.
Since solving $\PARITY$ requires $\Omg(n)$ queries, solving $\DLLCSRLE$ needs at least the same number of queries.
\end{proof}

\subsection{Lower Bounds on \texorpdfstring{$\ELLCSRLE$}{ELLCSRLE} and \texorpdfstring{$\LCSRLE$}{LCSRLE}}\label{sec:LB-LCSRLE}

In this section, we show an $\tOmg(n)$ lower bound on both $\ELLCSRLE$ (\autoref{thm:LB-ELLCSRLE}) and $\LCSRLE$ (\autoref{thm:LB-LCSRLE}).
More precisely, we reduce $\PARITY$ to $\ELLCSRLE$, which is then reduced to $\LCSRLE$.

Here is a high-level overview of the first reduction (from $\PARITY$ to $\ELLCSRLE$).
We encode an $n$-bit binary string $B=B_1B_2\cdots B_n$ into an RLE string $S_B$,
in a way similar to \autoref{eqn:SB} in the proof of \autoref{thm:LB-DLLCSRLE}.
We then assume an algorithm $\calA$ of query complexity $Q(\calA)$ for $\ELLCSRLE$ exists.
Using $\calA$, we construct an algorithm to compare the decoded length of $S_B$, \ie $|\tS_B|$, with any $k>0$.
We then use binary search on $k$ to find $|\tS_B|$, invoking $\calA$ for $\Od(\log n)$ times.
From $|\tS_B|$, we calculate the parity of $B$ without extra query.
Finally, since $\PARITY$ has query lower bound $\Omg(n)$, $Q(\calA)$ is at least $\tOmg(n)$, getting \autoref{thm:LB-ELLCSRLE} below.

\begin{lma}[Lower Bound of $\ELLCSRLE$]\label{thm:LB-ELLCSRLE}
	Any quantum oracle algorithm $\calA$ requires at least $\tOmg(n)$ queries to solve $\ELLCSRLE$, with probability at least $2/3$.
\end{lma}

\begin{proof}
Given an $n$-bit binary string $B=B_1B_2B_3\ldots B_n\in\set{0,1}^n$,
we can construct an RLE string
\[
	S_B
	= \ta^{2B_1+2}\tb^{2B_2+2}\ta^{2B_3+2}\tb^{2B_4+2}\ldots\gma^{2B_n+2}
,\]
where $\gma$ is $\ta$ if $n$ is odd, otherwise it is $\tb$.

For every positive natural number $k$, we can also construct an RLE string, simply by repeating another character:
\[
	S_k = \tc^k.
\]
We then concatenate $S_B$ and $S_k$ together with different characters in the middle, getting
\newcommand{\Sat}[1][k]{S_{B,\tat,{#1}}}
\newcommand{\Ssh}[1][k]{S_{B,\tsr,{#1}}}
\newcommand{\tSat}[1][k]{\wtilde{S}_{B,\tat,{#1}}}
\newcommand{\tSsh}[1][k]{\wtilde{S}_{B,\tsr,{#1}}}
\begin{equation}\label{eqn:StatStsr}
  	\Sat := S_B\tat^1S_k
	\quad\text{and}\quad
	\Ssh := S_B\tsr^1S_k.
\end{equation}

Let us check what we know about the LCS $\ts$ between $\tSat$ and $\tSsh$.
Firstly, $\tat$ and $\tsr$ are not in $\ts$ since none of them appears in $\tSat$ and $\tSsh$ at the same time.
Secondly, $\ts$ is a substring of $\tS_B$ or $\tS_k$, but not both,
because the character set of $\tS_B$, $\set{\ta,\tb}$, and the one of $\tS_k$, $\set{\tc}$, do not intersect.
Finally, $\ts$ is the ``longest" common substring, so it is the longest one among $\tS_B$ and $\tS_k$.

Now we assume the algorithm $\calA$ in \autoref{thm:LB-ELLCSRLE} exists,
and it has query complexity $Q(\calA)$.

Additionally, the success probability of $\calA$ can be boosted from constant to high probability with an extra logarithmic factor on its query complexity (\autoref{thm:whp}).

With $\Sat$ and $\Ssh$ as inputs, $\calA$ outputs
\begin{align}
    \calA(\Sat,\Ssh)
    &=
    \begin{cases}
   	 |S_k|,& |\tS_k| > |\tS_B| \\
   	 |S_B| \text{ or } |S_k|,& |\tS_k| = |\tS_B|\\
   	 |S_B|,& |\tS_k| < |\tS_B|
    \end{cases}\\
    &=
    \begin{cases}\label{eqn:calA_B}
   	 1,& k > |\tS_B|\\
   	 n \text{ or } 1,& k = |\tS_B|\\
   	 n,& k < |\tS_B|
    \end{cases}
.\end{align}
 
For a given $B$, we use $\calA_B({}\cdot{})$ as a shorthand for $\calA(\Sat[\,\cdot\ ],\Ssh[\,\cdot\ ])$ in the following text.
Note that when $k = |\tS_B|$, two answers ($n$ and $1$) are possible,
and we only assume $\calA$ outputs one of them.
Thus, $\calA_B(k)$ is non-deterministic when $k = |\tS_B|$.
We will resolve this issue with a property of binary search later.

To find $|\tS_B|$, we do a binary search on $k$ to find a $k'\in[2n, 4n]$, such that $\calA_B(k'-1)=n$ and $\calA_B(k')=1$.%
\footnote{The search range $[2n, 4n]$ comes from $2n \leq |\tS_B|=\sum_i (2B_i+2) \leq 4n$.}
In the binary search, $\calA_B$ will not be called with the same $k$ twice so it does not matter whether $\calA_B$ and the underlying $\calA$ are deterministic or not.
So from now on, we treat $\calA_B$ as if it were deterministic.

Since there are two possible outputs for $\calA_B(k)$ when $k=|\tS_B|$
(the middle case in \autoref{eqn:calA_B}),
each corresponds to a different result $k'$ for the binary search.
If $\calA_B(|\tS_B|)$ outputs $1$, we will get $k' = |\tS_B|$, the desired result.
But if $\calA_B(|\tS_B|)$ outputs $n$, we will get $k' = |\tS_B|+1$ instead.
We can detect if the latter one is the case from the parity of $k'$ because $ |\tS_B| = 2\sum_{i=1}^n (B_i+1) $ is always even, and thus we can correct the result accordingly.

With $|\tS_B|$ in hand, we then check if
\[
	\sum_{i=1}^n B_i
	=
	\frac12\sum_{i=1}^n \left(2B_i + 2\right) - n
	=
	\frac12|\tS_B| - n
\]
is odd or even to determine the parity of $B$.

Alternatively, we can XOR the lowest bit of $n$ with the second-lowest bit of $k'$, which is the same as the one of $|\tS_B|$, directly.
Then the result is the parity of $B$.
This allows us to avoid correcting $k'$ explicitly.

In total, we use $Q(\calA)\log^2n$ queries to solve $\PARITY$.
The logarithmic factors come from boosting $\calA$ to high probability and the binary search.
Finally, solving $\PARITY$ requires $\Omega(n)$ queries so we have
\[
    Q(\calA)\log^2n \in \Omg(n)
    \implies
    Q(\calA) \in \Omg(n/\log^2n) \in \tOmg(n),
\]
and \autoref{thm:LB-ELLCSRLE} follows.
\end{proof}

Furthermore, an algorithm solving $\LCSRLE$ outputs a triplet $(i_A,i_B,\ell)$, where $\ell$ is the encoded length of LCS between the inputs, which is also the answer to $\ELLCSRLE$.
\Ie $\ELLCSRLE$ can be reduced to $\LCSRLE$ with no extra query to the input strings.
As a result, $\LCSRLE$ shares the same query lower bound with $\ELLCSRLE$.
This gives the corollary below.

\begin{col}[Lower Bound of $\LCSRLE$]\label{thm:LB-LCSRLE}
Any quantum oracle algorithm $\calA$ requires at least $\tOmg(n)$ queries to solve $\LCSRLE$, with probability at least $2/3$.
\end{col}

\autoref{thm:LB-LCSRLE} is our motivation and justification to introduce the prefix-sum oracles.

\section*{Acknowledgement}
 This work is supported by NSTC QC project under Grant no.  111-2119-M-001-004- and   110-2222-E-007-002-MY3.

\printbibliography
\end{document}